\documentclass[%
 reprint,
 amsmath,amssymb,
 aps,
]{revtex4-2}

\usepackage{graphicx}
\usepackage{dcolumn}
\usepackage{bm}

\usepackage{algorithm}
\usepackage{algpseudocode}
\usepackage{booktabs}
\usepackage{braket}
\usepackage{amsmath}

\usepackage{array}
\usepackage{caption}
\usepackage{graphicx} 

\usepackage{mathtools}

\usepackage{threeparttable} 

\usepackage{amsmath, amssymb, amsthm} 

\newtheoremstyle{appendixstyle} 
  {12pt}       
  {12pt}       
  {\itshape}   
  {}           
  {\bfseries}  
  {.}          
  {.5em}       
  {}           

\theoremstyle{appendixstyle}
\newtheorem{theorem}{Theorem}[section] 

\begin{document}

\preprint{APS/123-QED}

\title{Statistical Quantum Mechanics of the Random Permutation Sorting System (RPSS): A Self-Stabilizing True Uniform RNG}


\author{Randy Kuang}
  \email{randy.kuang@quantropi.com}
\affiliation{%
Quantropi (Canada)\\
 1545 Carling Av., Suite 620, Ottawa,ON K1Z 8P9, Canada. 
}%

\date{\today}

\begin{abstract}
\noindent We introduce the \emph{Random Permutation Sorting System} (RPSS), a novel, software-defined framework that generates high-quality, uniform randomness by leveraging a statistical model with properties analogous to quantum ensembles. The RPSS operates as a self-contained entropy engine based on a fundamental pair of \textbf{conjugate observables}: the permutation count ($\hat{N}_p$) and the elapsed sorting time ($\hat{T}$). While the raw distributions of these observables are highly right-skewed and heavy-tailed, we provide a mathematical proof of their synchronous convergence to uniformity. The permutation count distributions follow a negative binomial model, and the observed modes for different repetition parameters \(m\) closely match the theoretical predictions, confirming the validity of the model. This convergence arises from the compounded effect of permutation counts on elapsed time and a process of \textbf{modular uniformization}, in which a vast number of microstates—stemming from combinatorial complexity and intrinsic system-level jitter—map onto a finite set of uniform symbols. This mechanism enables the RPSS to function as a \textbf{True Uniform Random Number Generator} (TURNG).
We implemented the RPSS as the \textbf{QPP-RNG}, a practical software realization that autonomously harvests entropy from a platform's intrinsic microarchitectural dynamics. Unlike prior benchmark-focused studies, this work provides a deep theoretical and experimental investigation of the underlying physical mechanism. Empirical validation confirms our predictions, demonstrating rapid entropy convergence and statistically uniform outputs, as assessed by \textbf{NIST SP 800-90B min-entropy} and \textbf{chi-squared statistics}, consistent with central-limit-theorem expectations.
The RPSS establishes a new class of \textbf{quantum-inspired eco-cryptosystems}, in which randomness is simultaneously harvested from unpredictable hardware behavior and amplified through a combinatorial process. This dual-source approach provides a compact, self-stabilizing entropy engine suitable for foundational applications in cryptography, blockchain protocols, and digital currency systems, offering a robust, platform-agnostic alternative to conventional entropy sources.
\end{abstract}


\maketitle

\section{Introduction}\label{sec:introduction}

Randomness is a fundamental resource across science and technology, underpinning statistical physics, cryptography, and quantum information science~\cite{herrero2017physical, menezes1996handbook}. In cryptography, the quality of randomness is critical: secret keys, nonces, and ephemeral parameters must be both unpredictable and statistically uniform~\cite{nist80090a}.

Random number generators (RNGs) are typically classified by their entropy sources. Pseudorandom number generators (PRNGs) expand a small seed deterministically~\cite{lehmer1951, xorshift, MacLaren1970, marsaglia2003xorshift}. True random number generators (TRNGs) exploit classical stochastic processes such as thermal noise~\cite{holman1997integrated}, avalanche noise~\cite{sunar2007true}, or oscillator jitter, with \emph{JitterEntropy} being a notable CPU-based implementation~\cite{mueller2022cpujitter}. Quantum random number generators (QRNGs) instead rely on intrinsic quantum indeterminacy, e.g., single-photon detection or vacuum fluctuation sampling~\cite{zhang2023qrng, qrng-Gabriel2010, qrng-Ma2016}. Both TRNGs and QRNGs, however, typically yield biased or skewed raw outputs, necessitating post-processing through extractors or whitening filters~\cite{nist80090b}. Operating systems aggregate these entropy sources into deterministic random bit generators (DRBGs), such as Linux \texttt{/dev/random} and Windows CNG, but these remain dependent on external reseeding.

Building on the framework of the \emph{Quantum Permutation Pad (QPP)}~\cite{kuang2020qpp, qpp-springer-kuang-2022}, which generalizes the one-time pad into the full permutation group of $n$-bit data, we recently developed QPP-RNG: a software-based TRNG harvesting entropy from microarchitectural fluctuations in commodity hardware~\cite{qpp-rng-sci-kuang-2025}. These fluctuations, measured in discrete system clock units (ticks or nanoseconds), arise from CPU pipeline jitter, cache misses, memory-access latency, and OS scheduling. By cycling through iterative permutations, QPP-RNG amplifies this entropy and produces statistically uniform outputs independently of external post-processing.

Empirical evaluations show that QPP-RNG consistently passes stringent statistical tests, including the NIST SP 800-90B IID assessments, NIST SP 800-22 suite, and ENT benchmarks, across diverse platforms. Yet, the mechanism by which skewed timing and permutation data converge to uniformity remained elusive.

To provide a clear and quantitative basis for describing the observed uniformity, we classify QPP-RNG outputs using the NIST SP 800-90B MCV (Most Common Value) min-entropy estimator.

Streams with min-entropy above 7.95 bits/byte are labeled as \textbf{Near} uniform, and those above 7.99 bits/byte as \textbf{Essential} uniform. Supporting statistics such as mean and serial correlation are reported for diagnostic purposes but do not determine classification (see Sec.~\ref{sec:experiments}).

To explain this convergence, we introduce the \emph{Random Permutation Sorting System (RPSS)}, a quantum-inspired conceptual model formalizing QPP-RNG as a Quasi-Superposition Quantum System (QSQS)~\cite{kuang-qsqs-pre-2025}. RPSS is characterized by a pair of conjugate-like observables: a discrete, deterministic variable (the permutation count $\hat{N}$) and a discrete, stochastic variable (the sorting time $\hat{T}$). The inherent uncertainty in $\hat{T}$ prevents simultaneous precise determination with $\hat{N}$, providing a natural entropy source. Crucially, the skewed raw distributions of $\hat{N}$ and $\hat{T}$ are mapped into computational symbols via \emph{modular projection}, a process in which multiple microstates collapse onto the same symbol, structurally amplifying entropy and yielding near-uniform outputs. This transformation effectively converts a conventional TRNG into a \emph{True Uniform RNG (TURNG)} without additional extractors.

RPSS thus demonstrates a dual mechanism: entropy is harvested from intrinsic system fluctuations and simultaneously amplified by permutation degeneracy. This physics-inspired formulation provides a foundation for compact, self-contained entropy sources with potential applications to post-quantum cryptography, blockchains, and digital currencies, reducing reliance on external entropy while ensuring both unpredictability and uniformity.

The remainder of the paper is organized as follows: Sec.~\ref{sec:model} formalizes the RPSS framework; Sec.~\ref{sec:experiments} presents experimental results and entropy convergence analyses; and Sec.~\ref{sec:conclusion} concludes.

\section{Theory and Model}
\label{sec:model}

We present a formal theoretical framework for the \emph{Random Permutation Sorting System} (RPSS), an evolution of the previously introduced Quasi-Superposition Quantum System (QSQS)~\cite{kuang2025qpp-rng-axiv}. While QSQS highlighted the quantum-like uncertainty inherent in RPSS, here we provide a comprehensive treatment grounded in statistical quantum mechanics and information theory. We model RPSS as a classical system that exhibits statistical quantum behavior by mapping deterministic computation and stochastic physical effects or system jitter into a probabilistic structure with properties analogous to those of quantum ensembles.

\emph{\textbf{Important Note}:} The bra-ket notation and terms like "state" and "collapse" are used throughout this work as a powerful \emph{analogical framework} to describe the statistical properties of the classical RPSS. This formalism is chosen for its expressive power in capturing non-commutative observables and preparational uncertainty. No quantum phenomena are involved.

This section is organized in two parts. First, we define the components of the RPSS and its measurable observables in the subsection \textbf{The RPSS System}. Second, we formalize the statistical dynamics of the sorting process, showing how it generates high-quality entropy from its combinatorial and physical properties in the subsection \textbf{Statistical Mechanism of the RPSS System} (\ref{sec:principle}).

\subsection{The RPSS System}

The RPSS is defined as a computational--physical system consisting of the following interacting components:

\begin{itemize}
    \item \textbf{Computing substrate.}
    The underlying classical computing platform, including CPU pipelines and execution units, caches, main memory, buses, network interfaces, and the OS scheduler. This substrate is a source of stochastic fluctuations arising from thermal noise, voltage variations, interrupt timing jitter, process migration, and variable memory latency. These physical effects translate into significant runtime variance for individual operations (e.g., a cache miss can increase memory access latency by orders of magnitude), which is a key source of entropy for the elapsed-time observable.

    \item \textbf{Data.}
    A disordered integer array of \(N\) elements, \(\{c_i\}\), which can be viewed as the result of applying a permutation to an ordered reference array \(\{a_i\}\).

    \item \textbf{Core mechanism.}
    A random permutation sorting procedure whose goal is to find the correct inverse permutation.

    \item \textbf{Physical observables.}
    Two measurable quantities naturally arise:
    \begin{enumerate}
        \item \(\hat{N}_p\): the permutation count in a sorting cycle,
        \item \(\hat{T}\): the elapsed time of the sorting cycle.
    \end{enumerate}
    Due to stochastic substrate fluctuations, even identical permutation processes with the same \(\hat{N}_p\) exhibit variability in \(\hat{T}\). This reflects a conjugate-like relationship:
    \[
        [\hat{N}_p, \hat{T}] \ne 0,
    \]
    indicating that these observables cannot be simultaneously determined with arbitrary precision.

    \item \textbf{System states and measurement.}
    The initial state of the system is the known, disordered array, denoted as \(\ket{\{c_i\}}\). The action of the random permutation sorting process, denoted by the stochastic operator \(\hat{p}^{-1}\), acts on this preparatory state. Due to the stochasticity of the substrate, this operation generates a probability distribution over all possible outcomes, represented as a superposition of eigenstates:
    \begin{equation}\label{eq:state-superposition}
        \begin{aligned}
            \ket{\{c_i\}} &= \sum_{n_p} \beta_{n_p} \ket{n_p, \{a_i\}} \\
            &= \sum_t \gamma_t \ket{t, \{a_i\}},
        \end{aligned}
    \end{equation}
    where \(\beta_{n_p}\) and \(\gamma_t\) are probability amplitudes satisfying \(|\beta_{n_p}|^2 = \Pr(n_p)\) and \(|\gamma_t|^2 = \Pr(t)\). The eigenstates are defined by:
    \begin{equation}
    \begin{aligned}
        \hat{N}_p \ket{n_p, \{a_i\}} &= n_p \ket{n_p, \{a_i\}}, \\
        \hat{T} \ket{t, \{a_i\}} &= t \ket{t, \{a_i\}}.
        \end{aligned}
    \end{equation}
    
    The completion of a sorting cycle constitutes a \emph{measurement}, causing the system to resolve into a specific outcome state \(\ket{n_p, t, \{a_i\}}\) with probability \(|\alpha_{n_p, t}|^2\), where \(\alpha_{n_p, t}\) is the joint probability amplitude. The marginal probabilities are given by:
    \begin{equation}\label{eq:marginals2}
        |\beta_{n_p}|^2 = \sum_t |\alpha_{n_p,t}|^2, \qquad |\gamma_t|^2 = \sum_{n_p} |\alpha_{n_p,t}|^2.
    \end{equation}
    This measurement process is not a quantum physical effect but a formal description of the system's transition from a distribution of potential outcomes to a single realized outcome.
\end{itemize}

RPSS is therefore a concrete framework in which deterministic computation and stochastic physical effects jointly generate entropy. This duality between structured permutation dynamics and physical noise forms the foundation of the RPSS entropy source.

\subsection{Statistical Mechanism of the RPSS System}\label{sec:principle}

We now formalize the core sorting process using the state-operator formalism. The central idea is to model the disordered array as the result of an unknown permutation, \(\hat{p}\), applied to an ordered array \(\{a_i\}\). The goal of the random permutation sorting procedure is to apply the stochastic operator \(\hat{p}^{-1}\) to recover the original ordered state.

For an array of length \(N\), the permutation space \(\mathcal{P}\) contains \(N!\) elements. The RPSS entropy source operates by constructing \(\hat{p}^{-1}\) through the sequential multiplication of random permutations \(\hat{p}_j \in \mathcal{P}\):
\begin{equation}\label{eq:qpp}
    \hat{p}^{-1} = \prod_{j=1}^{n_p} \hat{p}_j.
\end{equation}
Infinitely many such paths exist, differing in length \(n_p\), choice, and order of permutations, spanning an exponentially large space and introducing fundamental uncertainty. Uniformly random selection of \(\hat{p}_j\) is necessary for convergence; biased selection may induce cycles or stalling.

The action of the stochastic sorting process on the initial state defines a superposition of all possible outcomes. This process, which finds the inverse permutation $\hat{p}^{-1}$ by applying a sequence of random permutations, is represented as a transformation that maps the initial state to a probability distribution over the eigenstates of the observables:
\begin{equation}
\hat{p}^{-1} \ket{\{c_i\}} = \sum_{n_p, t} \alpha_{n_p, t} \ket{n_p, t, \{a_i\}}.
\end{equation}
The subsequent measurement, which occurs upon completion of the sorting cycle, collapses this superposition into a single outcome state:
\[
\ket{n_p, t, \{a_i\}} \quad \text{with probability} \quad |\alpha_{n_p, t}|^2.
\]
 with \(\sum_{n_p, t} |\alpha_{n_p, t}|^2 = 1\).

The system's evolution gives rise to a joint probability distribution over the observables \(\hat{N}_p\) and \(\hat{T}\). For cryptographic applications, we measure the modular-reduced observables:
\[
\tilde{N}_p = \hat{N}_p \bmod 2^n, \quad \tilde{T} = \hat{T} \bmod 2^n,
\]
which provide the final entropy output.

The noncommutativity expression \([\hat{N}_p, \hat{T}] \ne 0\) represents a mathematical analogy highlighting the conjugate-like relationship between these observables. This formalism captures the essential feature that the combinatorial randomness of the permutation count and the physical randomness of the elapsed time are statistically coupled and cannot be simultaneously determined with arbitrary precision.

The marginal distributions are given by:
\begin{equation}\label{eq:marginals}
    |\beta_{n_p}|^2 = \sum_t |\alpha_{n_p,t}|^2, \qquad
    |\gamma_t|^2 = \sum_{n_p} |\alpha_{n_p,t}|^2,
\end{equation}
and the system state can be expressed in either basis:
\begin{equation}\label{eq:dist}
    \begin{aligned}
        \hat{p}^{-1} \ket{\{c_i\}} &= \sum_{n_p} \beta_{n_p} \ket{n_p, \{a_i\}}, \\
        &= \sum_t \gamma_t \ket{t, \{a_i\}},
    \end{aligned}
\end{equation}
with the eigenstates defined by:
\begin{equation}
    \hat{N}_p \ket{n_p, \{a_i\}} = n_p \ket{n_p, \{a_i\}}, \quad
    \hat{T} \ket{t, \{a_i\}} = t \ket{t, \{a_i\}}.
\end{equation}

Introducing the computational $n$-bit basis, we define the modular-reduced eigenstates:
\begin{equation}
    \hat{\mathcal{N}}_p \ket{\tilde{n}_p, \{a_i\}} = \tilde{n}_p \ket{\tilde{n}_p, \{a_i\}}, \quad
    \tilde{n}_p = n_p \bmod 2^n,
\end{equation}
\begin{equation}
    \hat{\mathcal{T}} \ket{\tilde{t}, \{a_i\}} = \tilde{t} \ket{\tilde{t}, \{a_i\}}, \quad
    \tilde{t} = t \bmod 2^n.
\end{equation}
The state can then be represented in the computational basis as:
\begin{equation}\label{eq:dist-reduce}
    \begin{aligned}
        \hat{p}^{-1} \ket{\{c_i\}}
        &= \sum_{\tilde{n}_p=0}^{2^n-1} \tilde{\beta}_{\tilde{n}_p} \ket{\tilde{n}_p, \{a_i\}}, \\
        &= \sum_{\tilde{t}=0}^{2^n-1} \tilde{\gamma}_{\tilde{t}} \ket{\tilde{t}, \{a_i\}},
    \end{aligned}
\end{equation}
where the probability amplitudes \(\tilde{\beta}_{\tilde{n}_p}\) and \(\tilde{\gamma}_{\tilde{t}}\) define the modular-reduced distributions over the $n$-bit computational basis.

This operator-level uncertainty ensures that the stochasticity of \(n_p\) and \(T\) is fundamentally coupled, producing high-quality entropy that manifests in the synchronous convergence of the modular-reduced variables \(\tilde{N}_p\) and \(\tilde{T}\). The RPSS system exploits this joint uncertainty to guarantee that both observables contribute to a self-stabilizing, near-uniform output, forming the cornerstone of the self-converged TURNG behavior.

We will experimentally demonstrate that while raw distributions of \(\hat{N}_p\) and \(\hat{T}\) are generally right-skewed with long tails, their modular-reduced distributions become essentially uniform once the RPSS system reaches entropy convergence, typically when \(M = m \cdot N! \gg 2^n\).

\subsection{Statistical Laws of RPSS Observables}
\label{sec:statistics}

\subsubsection{Permutation Count}\label{sec:permutation-count}
We now analyze the distribution of the permutation count \(n_p\), defined as the number of random draws required for a QPP pad \(\hat{p}\) to generate the target inverse permutation operator \(\hat{p}^{-1}\). A QPP pad is constructed by sequentially multiplying random permutations \(\pi \in \mathcal{P}\), where \(\mathcal{P}\) is the permutation group of size \(N!\). At each trial, one permutation is sampled independently and uniformly from \(\mathcal{P}\).

Let \(M = m \cdot N!\) represent the effective counting for the \(m\)-th success. However, the success probability per trial is determined solely by the size of the permutation group:
\[
p = \frac{1}{N!}.
\]
This reflects the probability that a single permutation trial produces the inverse permutation \(\hat{p}^{-1}\).

\paragraph*{General Case: \(m\)-th Success}
The process halts upon the \(m\)-th success (i.e., the inverse permutation \(\hat{p}^{-1}\) is generated \(m\) times). The permutation count \(n_p\) follows a negative binomial distribution:
\begin{equation}
\Pr[n_p = k] = |\beta_{n_p}|^2 = \binom{k-1}{m-1} (1-p)^{k-m} p^m, \qquad k \geq m,
\end{equation}
with success probability \(p = 1/N!\).

The expectation and variance are:
\[
\mathbb{E}[n_p] = \frac{m}{p} = m \cdot N!, \qquad 
\mathrm{Var}(n_p) = \frac{m(1-p)}{p^2} \approx m \cdot (N!)^2.
\]

\paragraph*{Special Case: \(m=1\) (First Success)}
When the process halts upon the first success, the negative binomial distribution reduces to a geometric distribution:
\begin{equation}
\Pr[n_p = k] = |\beta_{n_p}|^2 = (1-p)^{k-1} p, \qquad k \geq 1.
\end{equation}
The expectation and variance simplify to:
\[
\mathbb{E}[n_p] = \frac{1}{p} = N!, \qquad 
\mathrm{Var}(n_p) = \frac{1-p}{p^2} \approx (N!)^2.
\]
In the continuous limit with \(N! \gg 1\) and \(p \ll 1\), the geometric law converges to an exponential distribution with rate parameter \(\lambda = 1/N!\):
\begin{equation}
\Pr[n_p \approx x] \sim \lambda e^{-\lambda x}, \qquad \lambda = \frac{1}{N!}.
\end{equation}
This theoretical result aligns with empirical observations (e.g., for \(N>2\), \(m=1\)), where the histogram of \(n_p\) decays approximately exponentially, with the first trial probability \(\Pr[n_p=1] = p = 1/N!\).

\paragraph*{Mode of Permutation Count Distribution}
In addition to the expectation and variance, we also examine the mode of the permutation count distribution. The mode corresponds to the value of \(n_p\) that occurs most frequently, i.e., the most likely number of trials required for the system to generate the inverse permutation \(\hat{p}^{-1}\).

For the general case of the \(m\)-th success (\(m > 1\)), the mode of the negative binomial distribution is given by
\[
n_p^{\text{mode}} = \left\lfloor \frac{(m-1)(1-p)}{p} \right\rfloor + 1,
\]
where \(p\) is the success probability per trial. In our RPSS system, each permutation has a probability of
\[
p = \frac{1}{N!}.
\]
Substituting this into the mode formula, we obtain
\begin{equation}\label{eq:mode}
n_p^{\text{mode}} =  (m-1)(N!-1)  + 1.
\end{equation}
This accurately predicts the most likely number of trials required for the \(m\)-th occurrence of the inverse permutation.

For the special case of \(m = 1\) (first success), the distribution reduces to a geometric distribution, whose mode is
\[
n_p^{\text{mode}} = 1,
\]
indicating that the most frequent number of trials is just the first trial. This aligns with the observed RPSS behavior, where for \(N=4\) the modes for \(m=2,3,4\) are observed as 24–25, 48–49, and so on, consistent with the theoretical predictions.

\paragraph*{Distributional Properties}
For both cases, the distribution is highly right-skewed when \(M \gg 1\):
\begin{itemize}
    \item For \(m=1\), most realizations of \(n_p\) cluster around \(M = N!\), with the probability mass decaying exponentially. The mode of the distribution is also approximately \(M\), reinforcing this behavior.
    \item For \(m>1\), the distribution broadens: realizations cluster around \(M = m N!\), while the probability mass decays more slowly for large values of \(n_p\). The mode shifts to \((m-1)M\), and the distribution becomes longer-tailed.
\end{itemize}

The unified parameter \(M = m \cdot N!\) serves as the characteristic scale of the distribution, determining both its mean, variance, and mode. This formulation emphasizes the continuous transition between the geometric (\(m=1\)) and negative binomial (\(m>1\)) cases, providing a consistent framework for analyzing the permutation count statistics in the RPSS system.

\subsubsection{Elapsed-Time Observable \(\hat{T}\)}
\label{sec:elapsed-time}

In the RPSS system, the elapsed time \(\hat{T}\) for a complete sorting cycle is the fundamental measure of performance. It represents the cumulative runtime over all permutations applied in a QPP pad until the \(m\)-th success. While the permutation count \(\hat{N}_p\) (Sec.~\ref{sec:permutation-count}) measures the number of \emph{attempts}, \(\hat{T}\) incorporates the \emph{duration} of each attempt, making it a more complex and practically significant observable.

Let \(X_j\) be an independent and identically distributed (i.i.d.) random variable representing the runtime, including system jitter, for the \(j\)-th permutation. The total elapsed time is then a random sum:
\begin{equation}\label{eq:t}
\hat{T} = \sum_{j=1}^{\hat{N}_p} X_j,
\end{equation}
where \(\hat{N}_p \sim \text{NegativeBinomial}(m, p)\) is the random permutation count. This structure defines \(\hat{T}\) as a \textbf{compound random variable}.

\paragraph*{Compound Distribution Model}
The conditional moments of \(\hat{T}\) follow directly from the i.i.d. nature of the \(X_j\):
\begin{align}
\mathbb{E}[\hat{T} \mid \hat{N}_p = k] &= k \mu_X, \\
\mathrm{Var}(\hat{T} \mid \hat{N}_p = k) &= k \sigma_X^2.
\end{align}
Applying the law of total expectation and variance unveils how the variability in both the count and the runtime contributes to the overall uncertainty:
\begin{align}
\mathbb{E}[\hat{T}] &= \mathbb{E}[\hat{N}_p] \mu_X = \frac{m}{p} \mu_X, \\
\mathrm{Var}(\hat{T}) &= \mathbb{E}[\hat{N}_p] \sigma_X^2 + \mathrm{Var}(\hat{N}_p) \mu_X^2 = \frac{m}{p} \sigma_X^2 + \frac{m(1-p)}{p^2} \mu_X^2.
\end{align}
The variance decomposition is crucial: the first term (\(\frac{m}{p} \sigma_X^2\)) is the average variability from the runtimes themselves, while the second term (\(\frac{m(1-p)}{p^2} \mu_X^2\)) is the additional variability introduced solely by the randomness of the count \(\hat{N}_p\).

\paragraph*{Distributional Characteristics and Contrast with \(\hat{N}_p\)}
The distribution of \(\hat{T}\) is not a simple negative binomial; it is the convolution of a negative binomial number of runtime distributions. Its shape is fundamentally governed by the distribution of the individual runtimes \(X_j\). This leads to a key empirical observation:

For \(m=1\), the distribution of \(\hat{N}_p\) is geometric and peaks sharply at \(k=1\). However, the distribution of \(\hat{T}\) does not peak at the theoretical minimum runtime. Instead, it peaks at a small but non-zero value (e.g., \(t=3\) ticks). This occurs because the runtime distribution \(X_j\) has its own most probable value greater than zero. The event \(\{\hat{N}_p=1, X_1 = 3\}\) has the highest joint probability, forming the observed peak. The distribution remains right-skewed, but its tail is a function of both the geometric tail of \(n_p\) and the tail of the runtime distribution.

As \(m\) increases, the distribution of \(\hat{N}_p\) broadens. Consequently, the distribution of \(\hat{T}\) evolves from the sharp, discrete-like shape observed for \(m=1\) into a smooth, right-skewed distribution. For \(m = 2, 3, 4\), \(\hat{T}\) exhibits a characteristic right skew, with a peak that shifts right in proportion to \(m\), consistent with the scaling of \(\mathbb{E}[\hat{T}]\). The convolution with the runtime distribution \(X_j\) smooths the distribution and extends its right tail, in agreement with the empirical histograms.

\paragraph*{Extreme Value Behavior and Worst-Case Analysis}
The tail probability \(\Pr(\hat{T} > t)\) is paramount for latency guarantees. The compound nature of \(\hat{T}\) means extreme values can arise from two distinct, but not mutually exclusive, mechanisms:
\begin{enumerate}
    \item \textbf{Many Attempts:} The event \(\{\hat{N}_p = k\}\) for a large \(k \gg \mathbb{E}[\hat{N}_p]\), which has a non-negligible probability under the negative binomial law. This contributes a tail that decays algebraically.
    \item \textbf{Exceptionally Long Runtimes:} The event that the sum of \(k\) runtimes is large, which can occur for any \(k\), including \(k \approx \mathbb{E}[\hat{N}_p]\).
\end{enumerate}
The overall tail is a convolution of these two effects, formalized by the lower bound for any \(k\):
\begin{equation}\label{eq:long-tail}
\Pr(\hat{T} > t) \geq \Pr(\hat{N}_p = k) \cdot \Pr\left(\sum_{j=1}^k X_j > t \right).
\end{equation}
For \(m=1\), the runtime distribution \(X_j\) is highly concentrated, which suppresses the tail of \(\hat{T}\) relative to \(\hat{N}_p\); the most likely way to get a large \(T\) is solely through a large \(n_p\). However, for \(m>1\), the distribution of \(\hat{T}\) develops a heavier tail than \(\hat{N}_p\) alone. This is because the convolution of the negative binomial distribution with the runtime distribution smooths and extends the tail, and, more importantly, the worst-case performance is driven by the \emph{synergy} of both mechanisms: the joint occurrence of a large number of permutations \emph{and} one or more long runtimes within that sequence. This "perfect storm" scenario, while rare, is the primary contributor to the extreme values observed in the empirical distribution of \(\hat{T}\) and must be accounted for in worst-case latency analysis.

\subsubsection{Modular Uniformization of RPSS Observables}
\label{sec:modular-unif}

Both permutation count $n_p$ and elapsed time $T$ exhibit complex, system-dependent distributions with significant right-skewness. However, for cryptographic applications, we can transform these raw observables into uniformly distributed $n$-bit symbols through modular reduction.

\paragraph*{Mathematical Framework}
Let $Y$ represent either observable ($n_p$ or $T$). For cryptographic output, we reduce $Y$ modulo $R = 2^n$ to obtain an $n$-bit symbol:
\begin{equation}
\tilde{Y} = Y \bmod R, \qquad 0 \le r < R.
\end{equation}
The probability of each residue is the sum of the probabilities of all real numbers in that equivalence class:
\begin{equation}
\Pr[\tilde{Y} = r] = \sum_{k=0}^{\infty} \Pr[Y = r + kR].
\end{equation}

\paragraph*{Uniformization Mechanism}
The modular reduction acts as an averaging operator over the distribution. When the characteristic scale $M$ of $Y$ is much larger than $R$ ($M \gg R$), the sum over all equivalence classes $r + kR$ for fixed $r$ samples the distribution at nearly constant intervals. If the distribution is slowly varying over scales of order $R$, this averaging produces approximately equal probabilities for all residues:
\begin{equation}
\Pr[\tilde{Y} = r] \approx \frac{1}{R}, \quad 0 \le r < R.
\end{equation}

\paragraph*{Application to Permutation Count}
For $\tilde{n}_p = n_p \bmod R$, the characteristic scale is $M = \mathbb{E}[n_p] = m \cdot N!$. Following the detailed derivation in Appendix~\ref{app:modular-derivation}, we obtain the asymptotic uniformity result:
\[
\Pr[\tilde{n}_p = r] = \frac{1}{R} + O\left(\frac{1}{M}\right), \qquad 0 \le r < R.
\]
The uniformity arises from averaging the slowly varying negative binomial tail over all equivalence classes modulo $R$.

\paragraph*{Application to Elapsed Time}
For $\tilde{T} = T \bmod R$, the analysis is more complex due to the compound nature of $T$. However, empirical results show that the same uniformization principle applies. When the expected value of $T$ is much larger than $R$, modular reduction produces approximately uniform residues:
\[
\Pr[\tilde{T} = t] \approx \frac{1}{2^n} + O\left(\frac{1}{\mathbb{E}[\hat{T}]}\right), \quad 0 \le t < 2^n.
\]
where \(\mathbb{E}[\hat{T}]\) is the mean of the elapsed time. This occurs because the distribution of $T$, while system-dependent, is typically slowly varying over scales larger than $R$ when $\mathbb{E}[T] \gg R$.

\paragraph*{Joint Uncertainty and Cryptographic Significance}
The modular reduction of both $\hat{N}_p$ and $\hat{T}$ ensures that their joint uncertainty translates into near-uniform distributions for $\tilde{n}_p$ and $\tilde{T}$. This joint uncertainty is the cornerstone of RPSS's self-converged True Uniform Random Number Generator (TURNG) behavior, enabling the system to achieve synchronous convergence and produce cryptographically strong, self-stabilizing random bits.

\paragraph*{The Dual Requirement: Scale and Shape}
The condition $M \gg R$ is necessary but not sufficient to ensure uniformity. The distribution must also be slowly varying over intervals of length $R$:
\begin{itemize}
    \item For $m=1$, the geometric distribution is highly skewed, requiring $M \gg R$ (large $N$)
    \item For $m>1$, the negative binomial distribution is more symmetric, allowing smaller $N$ with larger $m$
    \item For elapsed time $T$, the compound distribution is typically slowly varying when $\mathbb{E}[T] \gg R$
\end{itemize}

\paragraph*{Performance Considerations}
In practice, uniformity can be achieved by either:
\begin{enumerate}
\item Increasing $N$ (super-exponential increase in $M$)
\item Increasing $m$ (linear increase in $M$ with better distribution shape)
\end{enumerate}
For fixed computational cost, increasing $m$ is often more effective, as it improves the distribution shape while maintaining manageable permutation sizes.

\subsubsection{Entropy Analysis and Convergence of the RPSS}

The RPSS system generates entropy through stochastic permutation counts $\hat{N}_p$ and elapsed times \(\hat{T}\). The Shannon entropy associated with these observables in Eq.~\eqref{eq:dist-reduce} is
\begin{equation}
\begin{aligned}
    H(\hat{\mathcal{N}}_p) &= - \sum_{\tilde{n}_p} |\beta_{\tilde{n}_p}|^2 \log_2 |\beta_{\tilde{n}_p}|^2, \\
    H(\hat{\mathcal{T}}) &= - \sum_{\tilde{t}} |\gamma_{\tilde{t}}|^2 \log_2 |\gamma_{\tilde{t}}|^2.
\end{aligned}
\end{equation}
For an \(n\)-bit entropy source, both \( H(\hat{N}_p)\) and \(H(\hat{T})\) for modular-reduced samples must be very close to \(n\) bits.

\paragraph*{Entropy Convergence via $N$ and $m$}
Increasing the array size $N$ (hence $M = N!$) or the repetition parameter $m$ raises the raw source entropy:
\begin{equation}
\log_2(m \cdot N!) \gg n \quad \rightarrow \quad \log_2(m \cdot N!) > n + 2
\end{equation}
ensuring that the $n$-bit computational basis is well covered. This regime is referred to as \emph{entropy convergence}, where the system's raw distributions are sufficiently broad to produce nearly uniform modular-reduced outputs.

\paragraph*{Statistical Validation}

We validate entropy convergence of the RPSS system using three complementary criteria, with particular attention to the role of the Central Limit Theorem (CLT):

\begin{enumerate}
    \item \textbf{Min-entropy:}
    The modular-reduced distributions $\tilde{n}_p$ and $\tilde{T}$ satisfy the NIST SP 800-90B IID requirements. This confirms that the worst-case unpredictability is cryptographically sufficient, independent of higher-order distributional details.

    \item \textbf{Chi-square goodness-of-fit:}  
    The modular-reduced distributions are compared with the ideal uniform law. The observed \(\chi^2\) values converge toward their theoretical expectation as $N$ or $m$ increase, supporting the claim that sample frequencies approach uniformity.  

    \item \textbf{Central Limit Theorem (CLT) validation:}  
    The CLT does not apply to the raw observables $n_p$ and $T$, whose heavy-tailed distributions are far from normal. Instead, it applies to the \emph{modular-reduced} observables $\tilde{n}_p$ and $\tilde{T}$.  

    Under the uniform hypothesis, the frequency counts across $2^n$ bins should deviate from their expected value according to a Gaussian law centered at zero. Empirical deviations exhibit this bell-shaped distribution, with most falling within the $1\sigma$--$3\sigma$ bounds predicted by the CLT.  

    In particular:  
\begin{itemize}  
    \item For $\tilde{n}_p$, accumulated waiting times smooth the distribution; after modular reduction, residual fluctuations in bin counts are Gaussian, consistent with uniformity.  
    \item For $\tilde{T}$, the compound sum structure yields Gaussian deviations in modular-reduced bin counts, again supporting uniformity.  
\end{itemize}  
\end{enumerate}

Together, these three criteria—min-entropy, $\chi^2$ fit, and CLT-consistent Gaussian residuals—are used here as academic demonstrations to confirm the theoretical framework. Comprehensive benchmarking and stress testing have been reported previously, allowing this work to focus on the physical mechanism and its mathematical validation.

\subsection{Self-Converged TURNG}
\label{sec:turng}

The preceding sections established the statistical laws governing the RPSS observables \(\hat{N}_p\) and \(\hat{T}\), demonstrating that their modular reductions, \(\tilde{N}_p\) and \(\tilde{T}\), yield near-uniform \(n\)-bit distributions under entropy convergence (\(M \gg R\)). We now describe how these properties are leveraged to construct a self-sustaining, cryptographically strong random number generator: QPP-RNG as a TURNG. This model provides the theoretical foundation for the software implementation described in the introduction.

A fundamental assumption is that the Quantum Permutation Pad (QPP) is generated by a deterministic PRNG (e.g., an LCG) initialized with a secret random seed. The system operates in a cyclic fashion:
\begin{enumerate}
    \item \textbf{Entropy Injection:} The modular-reduced elapsed time \(\tilde{T}\) (system jitter) from the current sorting cycle is used to reseed the LCG. A lightweight, non-destructive operation preserves state continuity while injecting fresh entropy:
    \begin{equation}
    \text{seed}_{\text{new}} = (\text{seed}_{\text{old}} \ll k) + \tilde{T},
    \end{equation}
    where $\ll k$ denotes a bitwise shift.
    \item \textbf{Pad Generation:} The updated LCG generates the next QPP pad.
    \item \textbf{Sorting \& Output:} The new pad performs random permutation sorting on the disordered array. The resulting permutation count \(n_p\) is reduced modulo \(2^n\) to produce the final TURNG output:
    \begin{equation}
    \tilde{n}_p = n_p \bmod 2^n.
    \end{equation}
\end{enumerate}

The cornerstone of this construction is the \emph{synchronous convergence} of \(\tilde{N}_p\) and \(\tilde{T}\). The entropy from \(\tilde{T}\) ensures the continued unpredictability of the QPP pad sequence, which in turn guarantees that the output \(\tilde{n}_p\) remains uniformly distributed. This creates a positive feedback loop: the entropy source reinforces itself, leading to a \emph{self-converged} TURNG that stabilizes without external intervention and produces cryptographically suitable random bits.

\section{Experiments and Discussions}
\label{sec:experiments}
TURNG using RPSS has been benchmarked cross major computing platforms such as Windows with X86 processor, macOS with X86 process and ARM, and Rassberry Pi in our recent publication~\cite{qpp-rng-sci-kuang-2025} and mobile platforms: iOS and Android~\cite{kuang2025mobile}, against NIST SP 800-90B for IID. In this paper, we mainly focus on experimental demonstrations of its mechanism characteristics.

Experiments were performed on the following hardware and software platform summarized in Table~\ref{tab:platform}.

\begin{table}[H]
\centering
\caption{Experimental platform: an RPSS system}
\label{tab:platform}
\begin{tabular}{c c c}
\toprule
\textbf{Item}   & \textbf{Details} \\ \midrule
Device      & Device-L038-SMS \\
CPU         & Intel i5-1240P (1.7 GHz) \\
RAM         & 16 GB \\
System      & 64-bit, x64-based \\
OS          & Windows 11 Pro, 24H2 \\
OS Build    & 26100.4652 \\
Java Impl.  & QPP-RNG: RPSS model \\
\bottomrule
\end{tabular}
\end{table}

In this system, the permutation sorting time is in ticks. We will demonstrate the distributions for both permutation count and elapsed time to show their conjugated property in Section~\ref{sec:measure-conjugate}, their raw distribution transition against the sorting repetition \(m\) in Section~\ref{sec:raw-distributions}, the modular-reduced distribution in Section~\ref{sec:modular-disribution}, finally the true uniform random number generator or TURN in Section~\ref{sec:turng}.

\subsection{Measurements of the Conjugate Observables}
\label{sec:measure-conjugate}

The RPSS framework introduces two fundamental, \textbf{conjugate observables}: the permutation count, $n_p$, and the elapsed time, $T$. While $n_p$ represents the number of permutations required for the $m$-th success (here $m=1$), $T$ captures the actual runtime, which includes system-level jitter. Their conjugate relationship emerges from the intrinsic trade-off between computational effort ($n_p$) and temporal execution ($T$).

This relationship is empirically demonstrated in Table~\ref{tab:elapsed_times}. Each row corresponds to a specific QPP pad ($QPP_1$ to $QPP_{26}$) with a \textbf{deterministic} $n_p$. The five columns of elapsed time data show the unpredictable variability of system jitter across independent runs for the same set of pads.

\begin{table}[ht]
\centering
\caption{Elapsed time measurements for deterministic QPP pads with a disordered array \(\{3, 2, 0, 1\}\) and \(m=1\). The fourth elapsed time column demonstrates a dramatic slowdown, likely due to system resource contention.}
\label{tab:elapsed_times}
\footnotesize
\begin{tabular}{c c c c c c c}
\toprule
QPP Pad & $n_p$ & \multicolumn{5}{c}{Elapsed Time (ticks)} \\
\cmidrule(l){3-7}
 & & Run 1 & Run 2 & Run 3 & \textbf{Run 4*} & Run 5 \\
\midrule
$QPP_1$ & 5 & 22 & 19 & 22 & 124 & 20 \\
$QPP_2$ & 8 & 45 & 45 & 41 & 301 & 46 \\
$QPP_3$ & 3 & 15 & 15 & 15 & 106 & 16 \\
$QPP_4$ & 5 & 14 & 15 & 14 & 117 & 14 \\
$QPP_5$ & 77 & 297 & 226 & 273 & 1794 & 237 \\
$QPP_6$ & 5 & 15 & 17 & 14 & 654 & 21 \\
$QPP_7$ & 20 & 42 & 34 & 36 & 314 & 35 \\
$QPP_8$ & 8 & 19 & 17 & 19 & 172 & 17 \\
$QPP_9$ & 28 & 47 & 55 & 66 & 363 & 57 \\
$QPP_{10}$ & 45 & 69 & 80 & 89 & 518 & 74 \\
$QPP_{11}$ & 41 & 83 & 89 & 82 & 609 & 76 \\
$QPP_{12}$ & 46 & 91 & 89 & 118 & 520 & 99 \\
$QPP_{13}$ & 32 & 55 & 52 & 63 & 473 & 55 \\
$QPP_{14}$ & 23 & 117 & 48 & 39 & 297 & 37 \\
$QPP_{15}$ & 12 & 16 & 25 & 23 & 431 & 19 \\
$QPP_{16}$ & 6 & 9 & 12 & 13 & 79 & 11 \\
$QPP_{17}$ & 2 & 5 & 7 & 8 & 58 & 6 \\
$QPP_{18}$ & 8 & 10 & 16 & 15 & 92 & 15 \\
$QPP_{19}$ & 22 & 19 & 37 & 70 & 168 & 31 \\
$QPP_{20}$ & 15 & 12 & 29 & 25 & 123 & 30 \\
$QPP_{21}$ & 1 & 4 & 7 & 6 & 52 & 6 \\
$QPP_{22}$ & 46 & 35 & 69 & 45 & 287 & 69 \\
$QPP_{23}$ & 1 & 4 & 103 & 4 & 50 & 4 \\
$QPP_{24}$ & 11 & 10 & 18 & 16 & 107 & 18 \\
$QPP_{25}$ & 2 & 5 & 7 & 7 & 56 & 6 \\
$QPP_{26}$ & 6 & 8 & 9 & 10 & 78 & 53 \\
\bottomrule
\end{tabular}
\end{table}

This conjugation is highlighted by three key observations:
\begin{enumerate}
    \item \textbf{Stochastic Temporal Variability:} For a fixed pad (each row), the elapsed time varies significantly across runs. For example, QPP pad $QPP_6$ ($n_p=5$) shows a runtime ranging from 14 to 654 ticks, a nearly 47$\times$ difference.
    \item \textbf{Stochastic Algorithmic Variability:} Across a fixed run (each column), pads with similar algorithmic complexity (similar $n_p$) can have drastically different runtimes. For instance, the two pads with $n_p=5$ ($QPP_1$ and $QPP_6$) show very different execution times, differing by factors of 1.5--46.6 across the runs.
    \item \textbf{System-Algorithmic Coupling:} The dramatic slowdown in Run 4 demonstrates how system-level factors (e.g., CPU load, I/O delays) \textbf{conjugate} with the algorithmic complexity to produce the final, unpredictable runtime.
\end{enumerate}

This conjugate relationship fundamentally stems from the compound nature of the elapsed time observable:
\[
T = \sum_{j=1}^{n_p} X_j
\]
where $X_j$ is the runtime of individual permutations, which includes system-level stochasticity. The mathematical conjugation arises because $n_p$ and $T$ are related but non-commuting observables; precise knowledge of one necessarily introduces uncertainty in the other.

The observed variability underscores the cryptographic strength of the RPSS approach. Even with complete knowledge of the QPP pads and their permutation counts, an adversary cannot predict execution times due to this conjugate relationship. The fact that both modular-reduced $n_p$ and $T$ converge to uniform distributions---despite such dramatic variations---is the core principle enabling the self-convergent random number generation of the RPSS framework.

\begin{table}[ht]
\centering
\caption{Modular-reduced values ($\mod 16$) of permutation count and elapsed time. This transformation reveals the underlying uniformity and enhanced cryptographic utility of the conjugate observables.}
\label{tab:mod16_times}
\footnotesize
\begin{tabular}{c c c c c c c}
\toprule
QPP Pad & $\tilde{n}_p$ & \multicolumn{5}{c}{$\widetilde{\text{Elapsed Time}}$ ($\mod 16$)} \\
\cmidrule(l){3-7}
 & $(n_p \mod 16)$ & Run 1 & Run 2 & Run 3 & Run 4 & Run 5 \\
\midrule
$QPP_1$ & 5 & 6 & 3 & 6 & 12 & 4 \\
$QPP_2$ & 8 & 13 & 13 & 9 & 13 & 14 \\
$QPP_3$ & 3 & 15 & 15 & 15 & 10 & 0 \\
$QPP_4$ & 5 & 14 & 15 & 14 & 5 & 14 \\
$QPP_5$ & 13 & 9 & 2 & 1 & 2 & 13 \\
$QPP_6$ & 5 & 15 & 1 & 14 & 14 & 5 \\
$QPP_7$ & 4 & 10 & 2 & 4 & 10 & 3 \\
$QPP_8$ & 8 & 3 & 1 & 3 & 12 & 1 \\
$QPP_9$ & 12 & 15 & 7 & 2 & 3 & 9 \\
$QPP_{10}$ & 13 & 5 & 0 & 9 & 6 & 10 \\
$QPP_{11}$ & 9 & 3 & 9 & 2 & 1 & 12 \\
$QPP_{12}$ & 14 & 11 & 9 & 6 & 8 & 3 \\
$QPP_{13}$ & 0 & 7 & 4 & 15 & 9 & 7 \\
$QPP_{14}$ & 7 & 5 & 0 & 7 & 9 & 5 \\
$QPP_{15}$ & 12 & 0 & 9 & 7 & 15 & 3 \\
$QPP_{16}$ & 6 & 9 & 12 & 13 & 15 & 11 \\
$QPP_{17}$ & 2 & 5 & 7 & 8 & 10 & 6 \\
$QPP_{18}$ & 8 & 10 & 0 & 15 & 12 & 15 \\
$QPP_{19}$ & 6 & 3 & 5 & 6 & 8 & 15 \\
$QPP_{20}$ & 15 & 12 & 13 & 9 & 11 & 14 \\
$QPP_{21}$ & 1 & 4 & 7 & 6 & 4 & 6 \\
$QPP_{22}$ & 14 & 3 & 5 & 13 & 15 & 5 \\
$QPP_{23}$ & 1 & 4 & 7 & 4 & 2 & 4 \\
$QPP_{24}$ & 11 & 10 & 2 & 0 & 11 & 2 \\
$QPP_{25}$ & 2 & 5 & 7 & 7 & 8 & 6 \\
$QPP_{26}$ & 6 & 8 & 9 & 10 & 14 & 5 \\
\bottomrule
\end{tabular}
\end{table}

The profound conjugate relationship between $n_p$ and $T$ is further crystallized by examining their modular reductions, $\tilde{n}_p = n_p \mod 16$ and $\tilde{T} = T \mod 16$, as shown in Table~\ref{tab:mod16_times}. The operation of modular reduction acts as a cryptographic extractor, transforming the raw, system-dependent unpredictability into a near-uniform distribution over the residue classes $\mathbb{Z}_{16}$.

While the raw data in Table~\ref{tab:elapsed_times} shows some predictable correlation---e.g., larger $n_p$ generally leads to larger $T$---this relationship is completely obscured in the modular domain. The values of $\tilde{n}_p$ and $\tilde{T}$ exhibit no obvious correlation to each other or to their pre-image values. For instance:

\begin{itemize}
    \item A large $n_p$ value like 77 ($QPP_5$) reduces to $\tilde{n}_p=13$, while a small $n_p=3$ ($QPP_3$) reduces to $\tilde{n}_p=3$. Their reduced elapsed times are distributed across the entire range of values (0 to 15) with no apparent pattern.
    \item Identical $\tilde{n}_p$ values can produce wildly different $\tilde{T}$ values. Pads with $\tilde{n}_p=5$ ($QPP_1, 4, 6$) yield $\tilde{T}$ values of $\{6, 3, 6, 12, 4\}$, $\{14, 15, 14, 5, 14\}$, and $\{15, 1, 14, 14, 5\}$ respectively. The entropy is preserved and maximized.
\end{itemize}

This demonstrates the core mechanism of the RPSS TURNG: the conjugate observables $n_p$ and $T$ provide a dual source of entropy. Their joint state is high-dimensional and unpredictable, but it is the final modular reduction that `flattens' this high-entropy state into a uniform, cryptographic-grade output, breaking any residual linear relationships and ensuring the bits are truly random and independent.

\subsection{Raw Distributions of the Conjugate Observables}\label{sec:raw-distributions}

In Section~\ref{sec:model}, we modeled the RPSS system within a statistical quantum mechanics framework. In this model, random permutation sorting acts as a transformation that converts an initial disordered state \(\ket{\{c_i\}}\) into a stochastic superposition state. This state may be expressed either in terms of permutation count eigenstates \(\ket{n_p, \{a_i\}}\) or elapsed-time eigenstates \(\ket{t, \{a_i\}}\), with probability amplitudes characterized by generally right-skewed distributions.

Figure~\ref{fig:raw_n_p} presents the permutation count distribution produced by a deterministic pseudorandom number generator—specifically, a fixed-seed linear congruential generator (LCG PRNG)~\cite{kuang2025-QPP-RNG}. Our analysis is based on measurements of permutation counts and elapsed times (in clock ticks) collected over one million trials. The corresponding elapsed-time distribution is shown in Fig.~\ref{fig:raw_t}.

For \(m = 1\), the permutation count distribution follows an exponential decay with a pronounced peak corresponding to a single permutation required to complete the sorting cycle. The observed probability at this peak is 41.3\%, which closely matches the theoretical value of 41.67\%. As the number of permutations increases, the probability decays exponentially. When the sorting process is extended to the \(m\)-th successful sort, the distribution transitions into a right-skewed form well approximated by a negative binomial distribution. As \(m\) increases, the peak shifts to the right and the tail lengthens accordingly.

\begin{figure}[ht]
   \centering
   \includegraphics[scale=0.3]{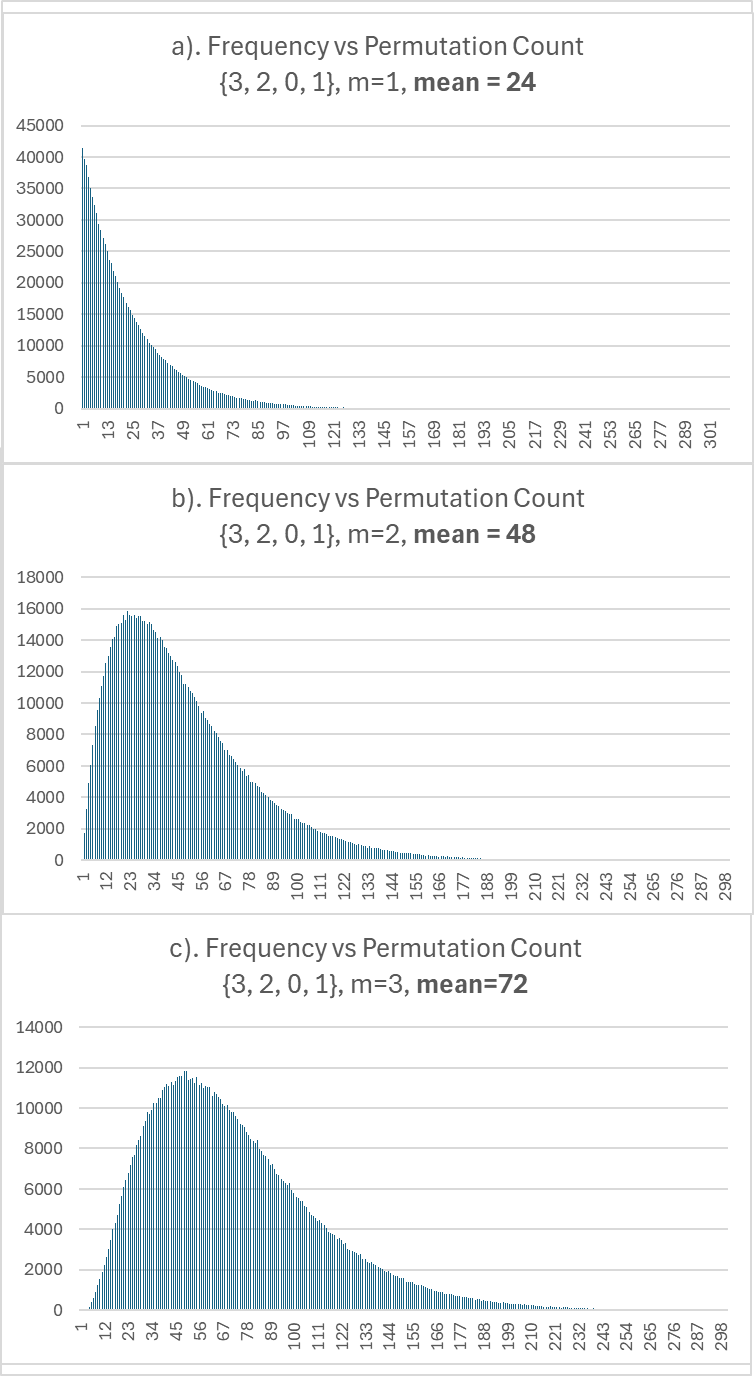}
   \caption{Raw permutation count distributions for sorting the disordered array \(\{3, 2, 0, 1\}\) at repetition values \(m = 1, 2, 3\).}
   \label{fig:raw_n_p}
\end{figure}

\begin{figure}[ht]
   \centering
   \includegraphics[scale=0.3]{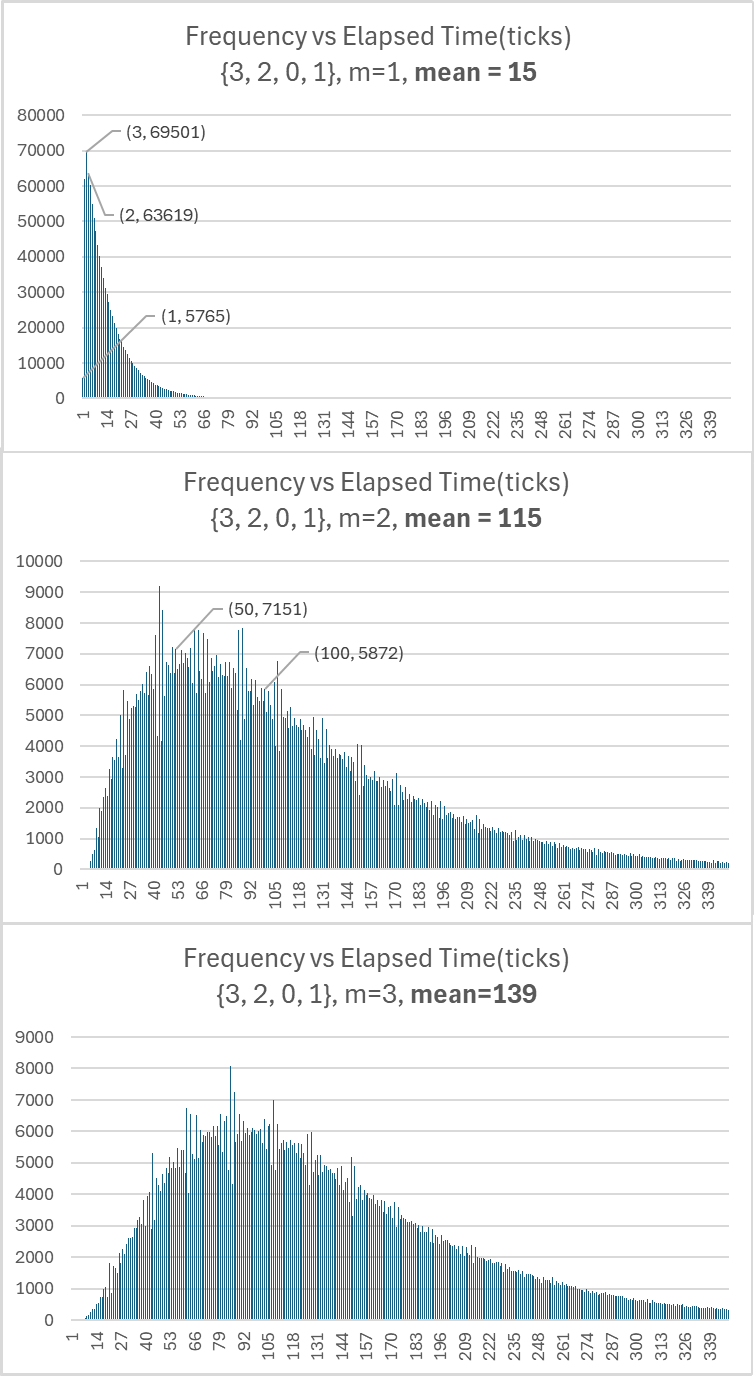}
   \caption{Raw elapsed-time distributions (in ticks) for sorting the disordered array \(\{3, 2, 0, 1\}\) at repetition values \(m = 1, 2, 3\).}
   \label{fig:raw_t}
\end{figure}

The elapsed-time distribution in Fig.~\ref{fig:raw_t} exhibits similar overall behavior to the permutation count distribution. For \(m = 1\), a sharp peak occurs at \(t = 3\) ticks, though non-negligible probabilities also appear at \(t = 1\) and \(2\) ticks. This deviation from the permutation count distribution reflects the compounded effects captured by the joint probability distribution in Eq.~\eqref{eq:t}. The mean elapsed time is 15 ticks, smaller than the mean permutation count \(M = mN! = 24\). Unlike the permutation count distribution—which is deterministic for fixed permutation pads—the elapsed-time distribution varies between runs due to system dynamics.

As \(m\) increases, both distributions undergo comparable transitions. The elapsed-time distributions are right-skewed and display sharp fluctuations indicative of runtime jitter. Compared to the permutation count distributions, the elapsed-time distributions are flatter and possess longer tails, consistent with the compounded contributions from all permutation counts, as predicted by Eq.~\eqref{eq:long-tail}. For example, the mean elapsed time is 115 ticks for \(m = 2\) and 139 ticks for \(m = 3\), both substantially larger than their corresponding permutation count means.

Extending to \(m=4\) produces a mean of \(M=mN!=96\). As expected, both the permutation count and elapsed-time distributions continue to exhibit right-skewed profiles, with their peaks and means shifting progressively rightward.

In our experiments for \(N=4\), the observed modes of the permutation count distributions were 1 for \(m=1\), 22 for \(m=2\), 48 for \(m=3\), and 71 for \(m=4\). The theoretical mode for \(m > 1\) is given by Eq.~\eqref{eq:mode}. Table~\ref{tab:modes} compares the observed and theoretical modes for \(m = 1, 2, 3, 4\):

\begin{table}[ht]
\centering
\caption{Comparison of observed and theoretical modes of the permutation count distribution for \(N=4\).}
\label{tab:modes}
\begin{tabular}{c c c}
\toprule
\(m\) & Observed Mode & Theoretical Mode (Eq.~\eqref{eq:mode}) \\
\midrule
1 & 1  & 1  \\
2 & 22 & 24 \\
3 & 48 & 47 \\
4 & 71 & 70 \\
\bottomrule
\end{tabular}
\end{table}

The small discrepancies for \(m=2\) and \(m=3\) are due to the right-skewness of the distribution and the discrete nature of \(n_p\), which slightly shifts the empirical peak relative to the theoretical value. Overall, these results confirm that the negative binomial framework accurately models the RPSS permutation count distribution, capturing the progressive rightward shift of the mode and the elongation of the tail as \(m\) increases.

\subsection{Modular-Reduced Distributions of the Conjugate Observables}\label{sec:modular-disribution}

We now examine the modular-reduced distributions for both permutation count and elapsed time, observing a remarkable transition toward uniform distributions for both conjugate observables.

As established in Section~\ref{sec:modular-unif}, when the mean $M \gg R = 2^n$ for an $n$-bit entropy generator, the modular-reduced distributions of both conjugate observables converge to uniformity. We demonstrate this phenomenon using the configuration $N = 4$ and $m = 3$, with results shown in Figs.~\ref{fig:mod_n_p} and~\ref{fig:mod_t}.

\begin{figure}[ht]
   \centering
   \includegraphics[scale=0.3]{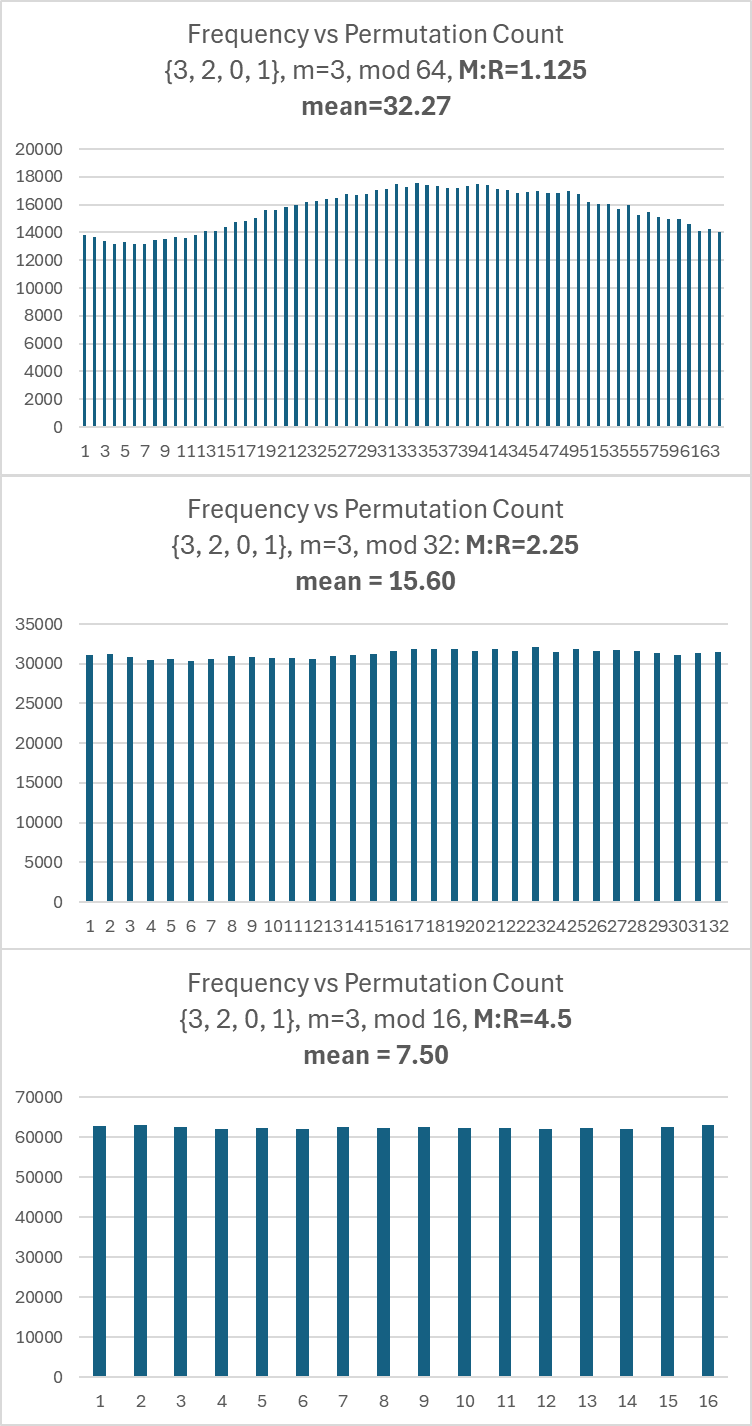}
   \caption{Modular-reduced permutation count distributions for sorting the disordered array $\{3, 2, 0, 1\}$ at repetition value $m = 3$.}
   \label{fig:mod_n_p}
\end{figure}

For $n = 6$, the modular-reduced permutation count distribution in Fig.~\ref{fig:mod_n_p} exhibits visually detectable non-uniformity across the 64 samples. This occurs because the ratio $\frac{M}{R} = 1.125$ does not satisfy the condition $M \gg R$. The same permutation counts demonstrate improved uniformity when reduced modulo 32 (i.e., for a 5-bit generator), though slight non-uniformity remains visible. In this case, the ratio $\frac{M}{R} = 2.25$ represents a significant improvement. The calculated arithmetic mean of 15.6 is very close to the ideal value of 15.50. For $m = 3$, the distribution approaches visual uniformity, and the arithmetic mean converges toward the ideal value of 7.50.

The modular-reduced elapsed time distribution in Fig.~\ref{fig:mod_t} shows similar behavior to that of the permutation count. Although the ratio $\frac{E(T)}{R} = 2.17$ is more favorable than for the permutation count, the distribution retains some non-uniformity due to the compounding effect of permutation counts. When projected onto a 5-bit basis ($n = 5$), the elapsed time distribution shows improved uniformity with an arithmetic mean of 15.49, closely matching the ideal value of 15.50, consistent with observations for the permutation count distribution. For a 4-bit generator, both elapsed time and permutation count distributions exhibit uniform characteristics.

\begin{figure}[ht]
   \centering
   \includegraphics[scale=0.3]{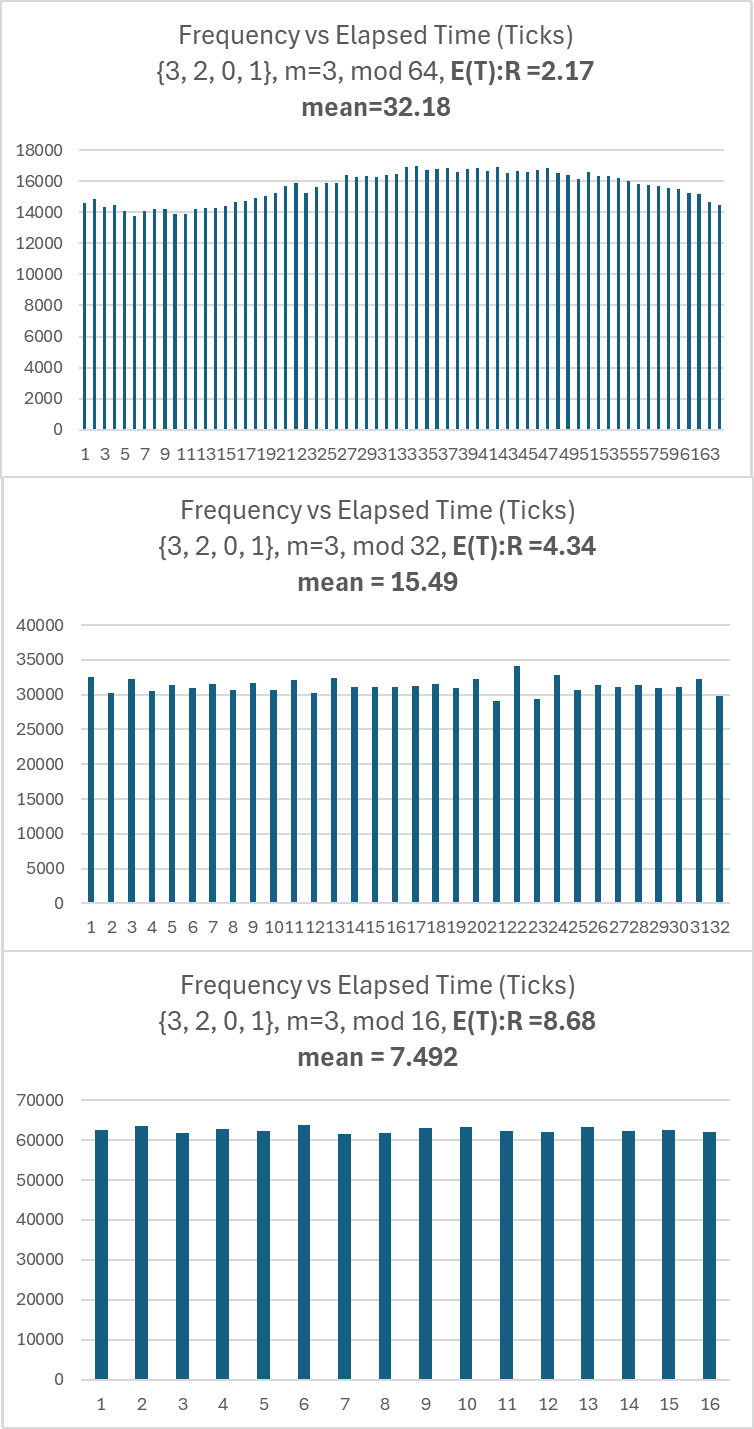}
   \caption{Modular-reduced elapsed-time distributions (in ticks) for sorting the disordered array $\{3, 2, 0, 1\}$ at repetition value $m = 3$.}
   \label{fig:mod_t}
\end{figure}

Thus, synchronous convergence to uniform distributions is achieved for $N = 4$ and $m = 3$. For enhanced uniformity, we recommend using $N = 4$ with $m \geq 4$. Table~\ref{tab:convergence} summarizes the convergence behavior for different values of $N$ and $m$.

\begin{table}[ht]
\centering
\caption{Convergence to uniformity for different array sizes ($N$) and repetition numbers ($m$).}
\label{tab:convergence}
\begin{tabular}{ccccl}
\hline
$N$ & $m$ & $M = m \cdot N!$ & $R = 2^n$ & Uniformity \\
\hline
3 & 15 & 90  & $n=4$ & Near \\
3 & 20 & 120 & $n=4$ & Essential \\
4 & 3  & 72  & $n=4$ & Near \\
4 & 4  & 96  & $n=4$ & Essential \\
5 & 2  & 240 & $n=4$ & Essential \\
5 & 4  & 480 & $n=8$ & Near    \\
5 & 5  & 600 & $n=8$ & Essential \\
\hline
\end{tabular}
\end{table}

\paragraph{Clarification of Near-Uniform and Essentially Uniform Behavior}

To remove ambiguity in terminology, we classify QPP-RNG outputs using the NIST SP 800-90B MCV min-entropy estimator. Since min-entropy lower-bounds Shannon entropy, it provides a conservative and reliable measure of uniformity. A stream is considered \textbf{essentially uniform} if its estimated min-entropy exceeds 7.99 bits/byte, and \textbf{near-uniform} if it exceeds 7.95 bits/byte. Supporting statistics, such as the arithmetic mean and serial correlation, are reported for diagnostic purposes but are not used for classification.

\begin{table}[ht]
\centering
\caption{Min-entropy-based uniformity metrics for the 4-element QPP-RNG array.}
\label{tab:uniformity_metrics}
\begin{tabular}{ccccc}
\hline
$m$ & $N_s$ & Elapsed Time & Permutation & Uniformity Class \\
\hline
4 & $10^6$  & 7.932 & 7.931 & Near \\
4 & $10^7$  & 7.977 & 7.981 & Essential \\
4 & $10^8$  & 7.994 & 7.994 & Essential \\
3 & $10^6$  & 7.905 & 7.933 & Near \\
3 & $10^7$  & 7.973 & 7.975 & Essential \\
3 & $10^8$  & 7.979 & 7.990 & Essential \\
\hline
\end{tabular}
\end{table}

Table~\ref{tab:uniformity_metrics} demonstrates that QPP-RNG outputs transition predictably from near-uniform to essentially uniform behavior as the permutation depth $m$ or sample size $N_s$ increases, providing a quantitative and reproducible framework for assessing uniformity.

\subsection{QPP-RNG as a Self-Training True Uniform RNG (TURNG)}\label{sec:tur}

In the preceding analysis, the RPSS system was examined under the constraint of a deterministic pseudorandom number generator—specifically, an LCG PRNG initialized with a fixed seed. Once synchronous convergence is attained under recommended configurations such as \((N=4, m=4)\), this constraint can be relaxed. At this stage, the conjugate observables—modular-reduced permutation counts \(\tilde{n}_p\) and elapsed times \(\tilde{t}\)—form a closed-loop feedback system that operates as a \emph{True Uniform Random Number Generator} (TURNG). A defining feature of TURNG is its capacity for autonomous seed evolution: irrespective of whether the initial seed is secret or public, the feedback mechanism iteratively diminishes its influence. After a finite number of training cycles, the RPSS state space retains no statistical memory of the initial seed, producing sequences of uniformly distributed random numbers.

The self-training mechanism functions analogously to a lightweight generative model. The RPSS dynamics serve as a generator, yielding stochastic outputs from intrinsic runtime fluctuations, while feedback from \(\tilde{t}\) acts as a discriminator that enforces uniformity and suppresses correlations. This generator–discriminator loop enables the system to adapt naturally to device-specific imperfections such as clock instabilities, timing jitter, and thermal drift. Consequently, TURNG is both reproducible—due to its foundation in the RPSS framework—and robust, as it leverages environmental fluctuations that vary across hardware.

Comprehensive benchmarking across servers, desktops, mobile platforms, and embedded devices~\cite{qpp-rng-sci-kuang-2025} confirms that conventional digital hardware can function as a high-quality entropy source without requiring specialized components. Conceptually, TURNG occupies an intermediate niche between deterministic PRNGs and hardware TRNGs: it is software-defined, platform-agnostic, and quantum-inspired, yet delivers the statistical qualities of a genuine random source. This framework also supports a distributed extension in which devices train locally while sharing only aggregate statistics to refine global parameters, suggesting a federated and scalable pathway for universal TURNG deployment.

\subsection{Comparison with Representative RNG Methods}

Table~\ref{tab:iid} places RPSS in context with representative TRNG and QRNG implementations whose min-entropy has been evaluated using comparable NIST SP~800-90B methodologies. While many modern RNGs successfully pass standard statistical test suites, these results illustrate that test pass rates alone do not fully capture architectural differences in entropy generation mechanisms.

\begin{table}[h]
\centering
\caption{Comparative Min-Entropy \(\mathbf{H_{\min}}\) Results}
\label{tab:iid}
\renewcommand{\arraystretch}{1.1}
\begin{tabular}{|l|l|c|c|}
\hline
\textbf{RNGs} & \textbf{Class} &  \text{bits/bit} & \text{bits/byte} \\ 
\hline
RPSS (this work) & TURNG & \textbf{0.9974} & \textbf{7.9792} \\
\hline
Red Hat RNG~\cite{nistE54}          & CPU Jitter & 0.9316 & 7.4528\\
Microchip~\cite{nist2023E46}       & TRNG & 0.5071 & 4.0568 \\
van der Waals~\cite{Abraham-TRNG-2022} & TRNG & 0.9830 & 7.8640 \\ 
Tunneling Effect~\cite{Liu-e-qrng-2023} & QRNG & 0.9872 & 7.8976 \\
IDQ Quantis~\cite{nist2020E63}     & QRNG & 0.9843 & 7.8744 \\
Quside PCIe~\cite{nist2024E178}    & QRNG & 0.8142 & 6.5136 \\
\hline
\end{tabular}
\end{table}

We emphasize that the reported min-entropy values for validated CPU jitter RNGs (e.g., Red Hat RNG~\cite{nistE54}) correspond to entropy estimates after conservative conditioning or entropy crediting at the module boundary, as required by NIST SP~800-90B validation. These values therefore reflect effective post-conditioning entropy rather than intrinsic raw-source entropy. In contrast, the RPSS/QPP-RNG results reported in Table~\ref{tab:iid} correspond to the native output of the generator and do not rely on external conditioning, extractors, or whitening stages.

Conventional CPU jitter–based RNGs rely on stochastic microarchitectural effects and typically require explicit conditioning or conservative entropy crediting at the module interface, since raw jitter measurements alone do not guarantee uniform symbol distributions. Physical TRNGs based on electronic or material noise sources (e.g., Microchip ECC608 or van der Waals heterostructures) demonstrate strong entropy generation but depend on specialized hardware and analog components. QRNGs achieve high min-entropy by exploiting intrinsic quantum indeterminacy, though at the cost of optical or solid-state quantum hardware and associated system complexity.

In contrast, RPSS achieves Near to Essential uniformity through an internal structural mechanism: entropy harvested from timing jitter is amplified by permutation degeneracy and modular projection, without reliance on external extractors or post-processing. This distinguishes RPSS from both software jitter RNGs and hardware-based TRNG/QRNG systems. As a result, RPSS offers a unique design point—a software-only entropy source capable of producing statistically uniform outputs with min-entropy comparable to leading physical and quantum RNGs, while remaining portable across commodity platforms.

From an engineering perspective, RPSS may be viewed as a form of entropy purification achieved through deterministic computational dynamics rather than statistical post-processing, with the observed uniformity emerging from structured aggregation of skewed timing statistics.

This comparison clarifies that RPSS is not intended to replace dedicated hardware RNGs in all settings, but rather to provide a complementary, self-contained alternative where deployability, architectural simplicity, and intrinsic uniformity are prioritized.

\section{Conclusion}
\label{sec:conclusion}

In this work, we introduced and formalized the \emph{Random Permutation Sorting System} (RPSS), a novel, software-defined framework that leverages intrinsic hardware fluctuations to generate high-quality entropy. By modeling the system with a quantum-inspired formalism, we demonstrated how the conjugate observables of permutation count ($\hat{N}_p$) and elapsed time ($\hat{T}$) capture the dual nature of our entropy source: combinatorial randomness from structured computation and physical randomness from microarchitectural jitter. We showed, both theoretically and empirically, that while the raw distributions of these observables are right-skewed and heavy-tailed, a simple modular reduction transforms them into statistically uniform outputs. This mechanism elevates the RPSS from a conventional true random number generator (TRNG) into a \emph{self-converged True Uniform Random Number Generator} (TURNG).

The empirical validation confirms our theoretical predictions, including the mode, expectation, and variance of the permutation count distributions, which follow the negative binomial model with success probability $p = 1/N!$. Table~\ref{tab:modes} compares the observed and theoretical modes for $N=4$ and $m=1$--$4$, demonstrating close agreement and confirming that the RPSS framework accurately captures both combinatorial and stochastic aspects of the system. The compound nature of the elapsed-time observable, which incorporates real-world system noise, is a powerful source of unpredictability. The modular uniformization process, in which a large number of microstates collapse onto a small set of computational symbols, acts as a cryptographic extractor, eliminating residual correlations and producing outputs that pass stringent statistical tests. Our analysis of entropy convergence as a function of array size ($N$) and repetition parameter ($m$) provides a practical guide for implementing a self-stabilizing entropy source on commodity hardware.

\paragraph{High-Quality Cryptographic Key Generation}  
Modern cryptography relies on strong entropy sources for generating secret keys, nonces, and ephemeral parameters. Current approaches often require specialized hardware or system-level aggregators, which can be expensive or vulnerable to attacks. The RPSS offers a robust alternative: it continuously generates cryptographically sound random bits from ordinary digital hardware. Its self-training mechanism ensures that randomness quality is maintained over time, regardless of the initial seed, guaranteeing resilience and long-term security.

\paragraph{Decentralized Ledger Technologies and Blockchains}  
Blockchains and distributed ledger technologies require verifiable, unpredictable, and decentralized randomness for applications such as leader election, fair lotteries, and decentralized applications (dApps). Unlike traditional TRNGs, which can be centralized and susceptible to manipulation, RPSS can be deployed on every node, enabling a federated model in which each device contributes to a global entropy pool by training locally and sharing only aggregate, non-sensitive data. This approach can facilitate truly decentralized random beacons, enhancing the security and integrity of distributed systems.

\paragraph{Foundations for Post-Quantum Cryptography and Eco-Cryptosystems}  
As quantum computing advances, the demand for quantum-resistant cryptographic primitives grows. While RPSS does not directly employ quantum phenomena, its conceptual foundation in quantum-like statistical uncertainty—specifically, the non-commutativity of its observables—offers a fresh perspective on designing random generators. Its reliance on physical, irreducible noise and the ability to produce complex joint probability distributions makes RPSS a promising building block for post-quantum random number generation.  

Moreover, RPSS can serve as the entropy engine of a \emph{conceptual post-quantum eco-cryptosystem}, embedding high-quality randomness directly into post-quantum cryptographic (PQC) protocols. By integrating RPSS with PQC primitives—such as lattice-based key encapsulation, hash-based signatures, or homomorphic encryption—the resulting eco-cryptosystem achieves a self-contained, self-stabilizing, and platform-agnostic security foundation. This approach unifies deterministic post-quantum cryptography with intrinsic physical entropy, providing both robustness against quantum-enabled attacks and resilience from environmental noise, while maintaining software-defined flexibility.

\paragraph*{RPSS as a Quantum-Statistical TURNG}
The Random Permutation Sorting System (RPSS) represents, to our knowledge, the first True Uniform Random Number Generator (TURNG) whose foundation is explicitly rooted in a \emph{quantum-inspired statistical mechanism}, while simultaneously harvesting \emph{physical entropy from commodity computing hardware}. Unlike conventional TRNGs that rely on specialized devices or post-processed noise, RPSS leverages the dual randomness of combinatorial permutation counts and microarchitectural jitter, formalized through conjugate observables. This duality ensures high-quality, self-converging uniform outputs without external intervention. In the quantum computing era, RPSS holds particular value: its physical, irreducible entropy is inherently unpredictable, offering a robust source of randomness for post-quantum cryptography. By embedding RPSS into quantum-resistant protocols or eco-cryptosystems, it is possible to create software-defined, self-stabilizing cryptographic primitives that remain secure even against adversaries equipped with quantum capabilities, thereby bridging classical hardware, statistical quantum theory, and next-generation secure computation.

In conclusion, the RPSS represents more than a new algorithm for generating random numbers; it embodies a fundamental rethinking of how high-quality entropy can be harvested from the pervasive fluctuations within classical computing devices. By providing a secure, reproducible, and robust software-defined entropy source, RPSS democratizes access to true randomness and lays the groundwork for next-generation post-quantum eco-cryptosystems, thereby fortifying the foundations of security and trust in an increasingly interconnected digital world.


\bibliography{my.bib}

@book{menezes1996handbook,
  title={Handbook of Applied Cryptography},
  author={Menezes, Alfred J. and van Oorschot, Paul C. and Vanstone, Scott A.},
  publisher={CRC press},
  year={1996}
}

@article{herrero2017physical,
  title={Physical random number generators},
  author={Herrero-Collantes, Miguel and Garcia-Escartin, Juan Carlos},
  journal={Reviews of Modern Physics},
  volume={89},
  number={1},
  pages={015004},
  year={2017}
}

@techreport{nist80090a,
  title={Recommendation for Random Number Generation Using Deterministic Random Bit Generators},
  author={National Institute of Standards and Technology (NIST)},
  number={SP 800-90A Rev.1},
  year={2015}
}

@techreport{nist80090b,
  title        = {Recommendation for the Entropy Sources Used for Random Bit Generation},
  author       = {S\"{o}nmez Turan, Meltem and Barker, Elaine and Kelsey, John M. and McKay, Kerry A. and Baish, Mary L. and Boyle, Mike},
  institution  = {National Institute of Standards and Technology},
  number       = {Special Publication 800-90B},
  year         = {2018},
  month        = jan,
  address      = {Gaithersburg, MD},
  doi          = {10.6028/NIST.SP.800-90B},
  url          = {https://nvlpubs.nist.gov/nistpubs/SpecialPublications/NIST.SP.800-90B.pdf}
}

@article{kuang-qsqs-pre-2025,
  title = {QPP-RNG: A conceptual quantum system for true randomness},
  author = {Kuang, Yurang (Randy)},
  journal = {Phys. Rev. E},
  volume = {112},
  issue = {5},
  pages = {055309},
  numpages = {12},
  year = {2025},
  month = {Nov},
  publisher = {American Physical Society},
  doi = {10.1103/d937-5yt7},
  url = {https://link.aps.org/doi/10.1103/d937-5yt7}
}

@misc{mueller2022cpujitter,
  author       = {Stephan M{\"u}ller},
  title        = {CPU Time Jitter Based Non-Physical True Random Number Generator},
  year         = {2022},
  howpublished = {\url{https://www.chronox.de/jent/CPU-Jitter-NPTRNG-v2.2.0.pdf}},
  note         = {Version 2.2.0 technical documentation of the Jitter RNG},
}

@article{sunar2007true,
  title={A provably secure true random number generator with built-in tolerance to active attacks},
  author={Sunar, Berk and Martin, Warren J. and Stinson, Doug R.},
  journal={IEEE Transactions on Computers},
  volume={56},
  number={1},
  pages={109--119},
  year={2007}
}

@article{holman1997integrated,
  title={An integrated analog/digital random noise source for random number generation},
  author={Holman, William T. and Foote, Matthew A. and Posey, Thomas W.},
  journal={IEEE Transactions on Circuits and Systems I: Fundamental Theory and Applications},
  volume={44},
  number={6},
  pages={521--528},
  year={1997}
}

@article{kuang2020qpp,
  author       = {Randy Kuang and Nicolas Bettenburg},
  title        = {Shannon Perfect Secrecy in a Discrete Hilbert Space},
  journal      = {IEEE International Conference on Quantum Computing and Engineering (QCE)},
  year         = {2020},
  pages        = {249--255},
  doi          = {10.1109/QCE49297.2020.00039}
}

@article{qpp-springer-kuang-2022,
    author = {Kuang,Randy and Barbeau, Michel},
    journal = {Quantum Information Processing},
    title = {Quantum Permutation Pad for Universal Quantum-Safe Cryptography},
    volume={21},
    pages={211},
    doi={10.1007/s11128-022-03557-y},
    year = {2022},
}

@article{kuang2025mobile,
  author       = {Randy Kuang and Jay Lou},
  title        = {QPP-RNG on Mobile Devices: Platform-Independent Entropy Harvesting via Permutation Sorting},
  journal      = {Quantum Information Processing},
  year         = {2025},
  note         = {submitted}
}

@INPROCEEDINGS{kuang2025-QPP-RNG,
  author={Kuang, Randy and Lou, Dafu},
  booktitle={2025 IEEE 14th International Conference on Communications, Circuits and Systems (ICCCAS)}, 
  title={IID-Based QPP-RNG: A Random Number Generator Utilizing Qpp Generated from System Jitter}, 
  year={2025},
  volume={},
  number={},
  pages={165-171},
  doi={10.1109/ICCCAS65806.2025.11102777}}

@article{qpp-rng-sci-kuang-2025,
  author    = {Vrana, G. and Lou, D. and Kuang, R.},
  title     = {Raw QPP-RNG randomness via system jitter across platforms: a NIST SP 800-90B evaluation},
  journal   = {Scientific Reports},
  volume    = {15},
  pages     = {27718},
  year      = {2025},
  doi       = {10.1038/s41598-025-13135-8},
  url       = {https://doi.org/10.1038/s41598-025-13135-8}
}

@inproceedings{lehmer1951,
  author       = {Lehmer, D. H.},
  title        = {Mathematical methods in large-scale computing units},
  booktitle    = {Proceedings of a Second Symposium on Large-Scale Digital Calculating Machinery},
  series       = {Annals of the Computation Laboratory of Harvard University},
  volume       = {26},
  pages        = {141--146},
  year         = {1951},
  publisher    = {Harvard University Press},
  address      = {Cambridge, MA}
}

@article{MacLaren1970,
  author    = {MacLaren, M. Donald},
  title     = {The Art of Computer Programming. Volume 2: Seminumerical Algorithms},
  journal   = {SIAM Review},
  volume    = {12},
  number    = {2},
  pages     = {306--308},
  year      = {1970},
  doi       = {10.1137/1012065},
  publisher = {Society for Industrial and Applied Mathematics},
  url       = {https://doi.org/10.1137/1012065},
  note      = {Accessed: 2025/02/14}
}

@article{marsaglia2003xorshift,
  author    = {Marsaglia, George},
  title     = {{Xorshift RNGs}},
  journal   = {Journal of Statistical Software},
  volume    = {8},
  number    = {14},
  pages     = {1--6},
  year      = {2003},
  doi       = {10.18637/jss.v008.i14}
}

@article{xorshift,
  author    = {Sebastiano Vigna},
  title     = {An experimental exploration of Marsaglia's xorshift generators, scrambled},
  journal   = {ACM Transactions on Mathematical Software},
  volume    = {42},
  number    = {4},
  pages     = {30:1--30:20},
  year      = {2016},
  publisher = {ACM},
  doi       = {10.1145/2845077},
  note      = {\url{https://vigna.di.unimi.it/ftp/papers/xorshift.pdf}}
}

@article{zhang2023qrng,
  title={Quantum Random Number Generation with Certified Min-Entropy Against Quantum Side Information},
  author={Zhang, Y. and Li, H.},
  journal={Physical Review Applied},
  volume={20},
  number={4},
  pages={044047},
  year={2023},
  publisher={APS}
}

@article{qrng-Ma2016,
  author      = {Ma, X. and Yuan, X. and Cao, Z. and others},
  title       = {Quantum random number generation},
  journal     = {npj Quantum Information},
  volume      = {2},
  pages       = {16021},
  year        = {2016},
  doi         = {10.1038/npjqi.2016.21},
  url         = {https://doi.org/10.1038/npjqi.2016.21}
}

@article{qrng-Gabriel2010,
  author      = {Gabriel, C. and Wittmann, C. and Sych, D. and Dong, R. and Mauerer, W. and Andersen, U. L. and Marquardt, C. and Leuchs, G.},
  title       = {A generator for unique quantum random numbers based on vacuum states},
  journal     = {Nature Photonics},
  volume      = {4},
  pages       = {711-715},
  year        = {2010},
  doi         = {10.1038/nphoton.2010.197}
}

@misc{kuang2025qpp-rng-axiv,
      title={QPP-RNG: A Conceptual Quantum System for True Randomness}, 
      author={Randy Kuang},
      year={2025},
      eprint={2508.01051},
      archivePrefix={arXiv},
      primaryClass={quant-ph},
      url={https://arxiv.org/abs/2508.01051}, 
}

@techreport{nist2020E63,
  title        = {{Entropy Source Validation Report for ID Quantique – Certificate E.63}},
  author       = {{NIST Cryptographic Module Validation Program}},
  institution  = {{National Institute of Standards and Technology}},
  year         = {2020},
  howpublished = {NIST CSRC Entropy Source Validation Public Use Document},
  url          = {https://csrc.nist.gov/CSRC/media/projects/cryptographic-module-validation-program/documents/entropy/E63_PublicUse.pdf},
  note         = {Accessed: 2026-01-08}
}

@techreport{nist2023E46,
  title        = {{Entropy Source Validation Report for Microchip ECC608 NRBG – Certificate E.46}},
  author       = {{NIST Cryptographic Module Validation Program}},
  institution  = {{National Institute of Standards and Technology}},
  year         = {2023},
  howpublished = {NIST CSRC Entropy Source Validation Public Use Document},
  url          = {https://csrc.nist.gov/CSRC/media/projects/cryptographic-module-validation-program/documents/entropy/E46_PublicUse.pdf},
  note         = {Accessed: 2026-01-08}
}

@techreport{nist2024E178,
  title        = {{Entropy Source Validation Report for Quside PCIe One – Certificate E.178}},
  author       = {{NIST Cryptographic Module Validation Program}},
  institution  = {{National Institute of Standards and Technology}},
  year         = {2024},
  howpublished = {NIST CSRC Entropy Source Validation Public Use Document},
  url          = {https://csrc.nist.gov/CSRC/media/projects/cryptographic-module-validation-program/documents/entropy/E178_PublicUse.pdf},
  note         = {Accessed: 2026-01-08}
}

@techreport{nistE54,
  title        = {{Entropy Source Validation Report for QuintessenceLabs Q5 RTX – Certificate E.54}},
  author       = {{NIST Cryptographic Module Validation Program}},
  institution  = {{National Institute of Standards and Technology}},
  year         = {2024},
  howpublished = {NIST CSRC Entropy Source Validation Public Use Document},
  url          = {https://csrc.nist.gov/CSRC/media/projects/cryptographic-module-validation-program/documents/entropy/E54_PublicUse.pdf},
  note         = {Accessed: 2025-05-08}
}

@Article{Liu-e-qrng-2023,
title = {A High-Randomness and High-Stability Electronic Quantum Random Number Generator without Post Processing},
journal = {Chin. Phys. Lett.},
volume = {40},
number = {7},
pages = {},
year = {2023},
issn = {},
doi = {10.1088/0256-307X/40/7/070303},	
url = {http://cpl.iphy.ac.cn/en/article/doi/10.1088/0256-307X/40/7/070303},
author = {Yu-Xuan Liu and Ke-Xin Huang and Yu-Ming Bai and Zhe Yang and Jun-Lin Li}
}

@article{Abraham-TRNG-2022,
  title   = {A High-Quality Entropy Source Using van der Waals Heterojunction for True Random Number Generation},
  author  = {Abraham, Nithin and Watanabe, Kenji and Taniguchi, Takashi and Majumdar, Kausik},
  journal = {ACS Nano},
  volume  = {16},
  number  = {4},
  pages   = {5898--5908},
  year    = {2022},
  doi     = {10.1021/acsnano.1c11084},
  publisher = {American Chemical Society}
}

\appendix
\section{Detailed Derivation for Modular Reduction Uniformity}
\label{app:modular-derivation}

This appendix provides rigorous proofs for the asymptotic uniformity of the modular-reduced permutation count $\tilde{n}_p$ and elapsed time $\tilde{T}$. The analysis for both observables relies on the fundamental scale parameter of the process:
\begin{equation}
M := \mathbb{E}[\hat{N}_p] = m \cdot N!.
\end{equation}

\subsection{Uniformity of the Modular-Reduced Permutation Count $\tilde{n}_p$}
\label{app:modular-unif-derivation}

\begin{theorem}[Uniformity of $\tilde{n}_p$]
\label{thm:np_uniformity}
Let $\hat{N}_p \sim \text{NegativeBinomial}(m, p)$ with mean $M = m/p$, and let $R = 2^n$ be a modulus. For any fixed integer $m \ge 1$ and fixed $R$, the modular reduction $\tilde{n}_p = \hat{N}_p \bmod R$ is asymptotically uniform as $M \to \infty$ (or equivalently, as $p \to 0$). That is, for any residue $0 \le r < R$,
\[
\Pr[\tilde{n}_p = r] = \frac{1}{R} + O\left(\frac{1}{M}\right).
\]
\end{theorem}

\begin{proof}
The probability of a residue $r$ is given by the sum over all equivalence classes:
\begin{equation}
\begin{aligned}
    S_r &:= \Pr[\tilde{n}_p = r] = \sum_{s=0}^{\infty} \Pr[\hat{N}_p = r + sR] \\
    &= \left(\frac{m}{M}\right)^m \sum_{s=0}^{\infty} \binom{r+sR-1}{m-1} \left(1 - \frac{m}{M}\right)^{r+sR-m}.
\end{aligned}
\end{equation}
Let $q := \left(1 - \tfrac{m}{M}\right)^R$. We proceed with an asymptotic expansion for large $M$.
Using the asymptotic form of the binomial coefficient,
\[
\binom{r+sR-1}{m-1} = \frac{(sR)^{m-1}}{(m-1)!} + O(s^{m-2}) \quad \text{as } s \to \infty,
\]
we can decompose the sum $S_r$ into a leading term and a remainder:
\begin{multline*}
S_r = \left(\tfrac{m}{M}\right)^m \left(1 - \tfrac{m}{M}\right)^{r-m} \times \\
\left[ \frac{R^{m-1}}{(m-1)!} \sum_{s=0}^{\infty} s^{m-1} q^s + O\left( \sum_{s=0}^{\infty} s^{m-2} q^s \right) \right].
\end{multline*}

The generating function for the polylogarithm gives:
\[
\sum_{s=0}^{\infty} s^{m-1} q^s = \frac{(m-1)!}{(1-q)^m} + O(1), 
\]
and 
\[
\quad \sum_{s=0}^{\infty} s^{m-2} q^s = O\left(\frac{1}{(1-q)^{m-1}}\right).
\]
Noting that $1 - q = \tfrac{mR}{M} + O(M^{-2})$, we substitute:
\begin{align*}
\frac{R^{m-1}}{(m-1)!} \sum_{s=0}^{\infty} s^{m-1} q^s
&= \frac{R^{m-1}}{(m-1)!} \cdot 
   \frac{(m-1)!}{(mR/M)^m}\,\bigl(1 + O(M^{-1})\bigr) \\
&= \frac{M^m}{m^m R}\,\bigl(1 + O(M^{-1})\bigr).
\end{align*}

The prefactor is $\left(\tfrac{m}{M}\right)^m \left(1 - \tfrac{m}{M}\right)^{r-m} = 1 + O(M^{-1})$.
Combining these results yields:
\begin{align*}
S_r 
&= \left(1 + O(M^{-1})\right) 
   \left[ \frac{1}{R}\bigl(1 + O(M^{-1})\bigr) + O\!\left(\tfrac{1}{M}\right) \right] \\
&= \frac{1}{R} + O\!\left(\tfrac{1}{M}\right).
\end{align*}
which completes the proof.
\end{proof}

\subsection{Uniformity of the Modular-Reduced Elapsed Time $\tilde{T}$}
\label{app:modular-unif-time}

\begin{theorem}[Uniformity of $\tilde{T}$]
\label{thm:t_uniformity}
Let $T = \sum_{j=1}^{\hat{N}_p} X_j$ be the elapsed time, where $\hat{N}_p$ is as in Theorem~\ref{thm:np_uniformity} and the $\{X_j\}$ are i.i.d. runtimes with mean $\mu_X$ and finite variance. Let $\mathbb{E}[T] = M \mu_X$ be the mean elapsed time.
Then, for a fixed modulus $R$, the modular reduction $\tilde{T} = T \bmod R$ is asymptotically uniform as $\mathbb{E}[T] \to \infty$:
\[
\Pr[\tilde{T} = t] = \frac{1}{R} + o(1) \quad \text{for } 0 \le t < R.
\]
A stronger result holds if the runtime distribution is non-lattice, yielding an error term of $O(1/\mathbb{E}[T])$.
\end{theorem}

\begin{proof}
The distribution of $\tilde{T}$ is given by the wrapped density:
\[
\Pr[\tilde{T} = t] = \sum_{k=0}^{\infty} f_T(t + kR), \quad 0 \le t < R,
\]
where $f_T$ is the probability density function of $T$.

The key insight is that for large $\mathbb{E}[T]$, the distribution $f_T(\tau)$ becomes increasingly concentrated and smooth around its mean $\mathbb{E}[T] \gg R$. By the Poisson summation formula,
\[
\sum_{k=-\infty}^{\infty} f_T(t + kR) = \frac{1}{R} \sum_{\ell=-\infty}^{\infty} \phi_T\left(\frac{2\pi\ell}{R}\right) e^{-i2\pi\ell t/R},
\]
where $\phi_T(\omega) = \mathbb{E}[e^{i\omega T}]$ is the characteristic function of $T$.

The $\ell=0$ term gives the uniform density $1/R$. For $\ell \neq 0$, we have:
\[
\left| \phi_T\left(\frac{2\pi\ell}{R}\right) \right| = \left| \phi_{N_p}\left( -i \ln \phi_X\left(\frac{2\pi\ell}{R}\right) \right) \right|,
\]
where $\phi_X$ is the characteristic function of the runtime.

Since $\mathbb{E}[T] = M \mu_X \to \infty$, and $T$ is a sum of a large number of random variables, $|\phi_T(\omega)|$ decays rapidly for $\omega \neq 0$. Specifically, for any fixed $\ell \neq 0$, $\phi_T(2\pi\ell/R) \to 0$ as $\mathbb{E}[T] \to \infty$. This convergence is exponential if the runtime distribution is non-lattice. Therefore, all terms for $\ell \neq 0$ vanish in the limit, leaving only the uniform term $1/R$.

The error bound $O(1/\mathbb{E}[T])$ follows from a more detailed analysis of the characteristic function's decay rate, which is inversely related to the variance of $T$, itself proportional to $M$ and thus $\mathbb{E}[T]$.
\end{proof}

Theorems~\ref{thm:np_uniformity} and~\ref{thm:t_uniformity} provide the mathematical foundation for the RPSS cryptographic output generation, demonstrating that modular reduction transforms the complex distributions of $\hat{N}_p$ and $\hat{T}$ into nearly uniform outputs when the system's inherent scale parameter $M$ is sufficiently large.
\end{document}